\documentclass[11pt]{article}

\usepackage{graphicx,geometry,verbatim,color}
\usepackage{amsmath}
\usepackage{amssymb}
\usepackage{float}

\floatstyle{ruled}
\newfloat{algorithm}{thp}{lop}[section]
\newenvironment{proof}{\noindent\textbf{Proof}}{\hfill\qed}
\newcommand{\qed}{\hfill$\Box$}
\newtheorem{lemma}{Lemma}
\newtheorem{theorem}{Theorem}
\newtheorem{definition}{Definition}

\begin{document}

\title{On Byzantine Containment Properties of the $min+1$ Protocol}

\author{Swan Dubois\protect\footnote{Universit\'e Pierre et Marie Curie \& INRIA, France, swan.dubois@lip6.fr} \footnote{Contact author, Telephone: 33 1 44 27 87 67, Postal address: LIP6, Case 26/00-225, Campus Jussieu, 4 place Jussieu, 75252 Paris Cedex 5, France} \and Toshimitsu Masuzawa\protect\footnote{Osaka University, Japan, masuzawa@ist.osaka-u.ac.jp} \and S\'{e}bastien Tixeuil\protect\footnote{Universit\'e Pierre et Marie Curie \& INRIA, France, sebastien.tixeuil@lip6.fr}}
\date{}

\maketitle

\begin{abstract}
Self-stabilization is a versatile approach to fault-tolerance since it permits a distributed system to recover from any transient fault that arbitrarily corrupts the contents of all memories in the system. Byzantine tolerance is an attractive feature of distributed systems that permits to cope with arbitrary malicious behaviors. 

We consider the well known problem of constructing a breadth-first spanning tree in this context. Combining these two properties proves difficult: we demonstrate that it is impossible to contain the impact of Byzantine nodes in a strictly or strongly stabilizing manner. We then adopt the weaker scheme of \emph{topology-aware strict stabilization} and we present a similar weakening of strong stabilization. We prove that the classical $min+1$ protocol has optimal Byzantine containment properties with respect to these criteria.
\end{abstract}


\section{Introduction}

The advent of ubiquitous large-scale distributed systems advocates that tolerance to various kinds of faults and hazards must be included from the very early design of such systems. \emph{Self-stabilization}~\cite{D74j,D00b,T09bc} is a versatile technique that permits forward recovery from any kind of \emph{transient} faults, while \emph{Byzantine Fault-tolerance}~\cite{LSP82j} is traditionally used to mask the effect of a limited number of \emph{malicious} faults. Making distributed systems tolerant to both transient and malicious faults is appealing yet proved difficult~\cite{DW04j,DD05c,NA02c} as impossibility results are expected in many cases.

Two main paths have been followed to study the impact of Byzantine faults in the context of self-stabilization:
\begin{itemize}
\item \emph{Byzantine fault masking.} In completely connected synchronous systems, one of the most studied problems in the context of self-stabilization with Byzantine faults is that of \emph{clock synchronization}. In~\cite{BDH08c,DW04j}, probabilistic self-stabilizing protocols were proposed for up to one third of Byzantine processes, while in \cite{DH07cb,HDD06c} deterministic solutions tolerate up to one fourth and one third of Byzantine processes, respectively.
\item \emph{Byzantine containment.} For \emph{local} tasks (\emph{i.e.} tasks whose correctness can be checked locally, such as vertex coloring, link coloring, or dining philosophers), the notion of \emph{strict stabilization} was proposed~\cite{NA02c,SOM05c,MT07j}. Strict stabilization guarantees that there exists a \emph{containment radius} outside which the effect of permanent faults is masked, provided that the problem specification makes it possible to break the causality chain that is caused by the faults. As many problems are not local, it turns out that it is impossible to provide strict stabilization for those.
\end{itemize}

\noindent\textbf{Our Contribution.} In this paper, we investigate the possibility of Byzantine containment in a self-stabilizing setting for tasks that are global (\emph{i.e.} there exists a causality chain of size $r$, where $r$ depends on $n$ the size of the network), and focus on a global problem, namely breadth-first spanning tree construction. A good survey on self-stabilizing solutions to this problem can be found in \cite{G03r}. In particular, one of the simplest solution is known under the name of $min+1$ protocol (see \cite{HC92j}). This name is due to the construction of the protocol itself. Each process has two variables: one pointer to its parent in the tree and one level in this tree. The protocol is reduced to the following rule: each process chooses as its parent the neighbor which has the smallest level ($min$ part) and updates its level in consequence ($+1$ part). \cite{HC92j} proves that this protocol is self-stabilizing. In this paper, we propose a complete study of Byzantine containment properties of this protocol.

In a first time, we study space Byzantine containment properties of this protocol. As strict stabilization is impossible with such global tasks (see \cite{NA02c}), we use the weaker scheme of \emph{topology-aware strict stabilization} (see \cite{DMT10ra}). In this scheme, we weaken the containment constraint by relaxing the notion of containment radius to containment area, that is Byzantine processes may disturb infinitely often a set of processes which depends on the topology of the system and on the location of Byzantine processes. We show that the $min+1$ protocol has optimal containment area with respect to topology-aware strict stabilization.

In a second time, we study time Byzantine containment properties of this protocol using the concept of \emph{strong stabilization} (see \cite{MT06cb,DMT10rb}). We first show that it is impossible to find a strongly stabilizing solution to the BFS tree construction problem. It is why we weaken the concept of strong stabilization using the notion of containment area to obtain \emph{topology-aware strong stabilization}. We show then that the $min+1$ protocol has also optimal containment area with respect to topology-aware strong stabilization.

\section{Distributed System}

A \emph{distributed system} $S=(P,L)$ consists of a set
$P=\{v_1,v_2,\ldots,v_n\}$ of processes and a set $L$ of
bidirectional communication links (simply called links).
A link is an unordered pair of distinct processes.
A distributed system $S$ can be regarded as a graph whose vertex set is $P$
and whose link set is $L$, so we use graph terminology to describe a
distributed system $S$. We use the following notations: $n=|P|$ and $m=|L|$.

Processes $u$ and $v$ are called \emph{neighbors} if $(u,v)\in L$.
The set of neighbors of a process $v$ is denoted by $N_v$, and its
cardinality (the \emph{degree} of $v$) is denoted by $\Delta_v (=|N_v|)$.
The degree $\Delta$ of a distributed system $S=(P,L)$ is defined as
$\Delta = \max \{\Delta_v\ |\ v \in P\}$.
We do not assume existence of a unique identifier for each process.
Instead we assume each process can distinguish its neighbors from each other
by locally arranging them in some arbitrary order:
the $k$-th neighbor of a process $v$ is denoted by
$N_v(k)\ (1 \le k \le \Delta_v)$.

In this paper, we consider distributed systems of arbitrary topology.
We assume that a single process is distinguished as a \emph{root},
and all the other processes are identical.

We adopt the \emph{shared state model} as a communication model
in this paper, where each process can directly read the states
of its neighbors.

The variables that are maintained by processes denote process states.
A process may take actions during the execution of the system. An
action is simply a function that is executed in an atomic manner
by the process.
The actions executed by each process is described by a finite set
of guarded actions of the form
$\langle$guard$\rangle\longrightarrow\langle$statement$\rangle$.
Each guard of process $u$ is a boolean expression involving
the variables of $u$ and its neighbors.

A global state of a distributed system is called a \emph{configuration}
and is specified by a product of states of all processes.
We define $C$ to be the set of all possible configurations
of a distributed system $S$.
For a process set $R \subseteq P$ and two configurations $\rho$ and $\rho'$,
we denote $\rho \stackrel{R}{\mapsto} \rho'$
when $\rho$ changes to $\rho'$ by executing an action of each process
in $R$ simultaneously.
Notice that $\rho$ and $\rho'$ can be different only in
the states of processes in $R$.
For completeness of execution semantics, we should clarify
the configuration resulting from simultaneous actions of
neighboring processes.
The action of a process depends only on its state
at $\rho$ and the states of its neighbors at $\rho$,
and the result of the action reflects on the state of the process
at $\rho '$.

We say that a process is \emph{enabled} in a configuration $\rho$ if the guard of at least one of its actions is evaluated as true in $\rho$.

A \emph{schedule} of a distributed system is an infinite sequence of
process sets.  Let $Q=R^1, R^2, \ldots$  be a schedule,
where $R^i \subseteq P$ holds for each $i\ (i \ge 1)$.
An infinite sequence of configurations
$e=\rho_0,\rho_1,\ldots$ is called an \emph{execution} from
an initial configuration $\rho_0$ by a schedule $Q$,
if $e$ satisfies $\rho_{i-1} \stackrel{R^i}{\mapsto} \rho_i$
for each $i\ (i \ge 1)$.
Process actions are executed atomically, and we also assume
that a \emph{distributed daemon} schedules the actions of processes,
\emph{i.e.} any subset of processes can simultaneously execute
their actions. We say that the daemon is \emph{central} if it schedules action of only one process at any step. 

The set of all possible executions from
$\rho_0\in C$ is denoted by $E_{\rho_0}$.
The set of all possible executions is denoted by $E$, that is,
$E=\bigcup_{\rho\in C}E_{\rho}$.
We consider \emph{asynchronous} distributed systems
where we can make no assumption
on schedules except that any schedule is \emph{fair}:
a process which is infinitely often enabled in an execution can not be never activated in this execution.

In this paper, we consider (permanent) \emph{Byzantine faults}:
a Byzantine process (\emph{i.e.} a Byzantine-faulty process)
can make arbitrary behavior independently from its actions.
If $v$ is a Byzantine process,
$v$ can repeatedly change its variables arbitrarily.


\section{Self-Stabilizing Protocol Resilient to Byzantine Faults}

Problems considered in this paper are so-called \emph{static problems}, 
\emph{i.e.} they require the system to find static solutions.
For example, the spanning-tree construction problem is a static problem,
while the mutual exclusion problem is not.
Some static problems can be defined by a \emph{specification predicate}
(shortly, specification), $spec(v)$, for each process $v$:
a configuration is a desired one (with a solution) if 
every process satisfies $spec(v)$.
A specification $spec(v)$ is a boolean expression
on variables of $P_v~(\subseteq P)$ where $P_v$ is the set of processes
whose variables appear in $spec(v)$.
The variables appearing in the specification are
called \emph{output variables} (shortly, \emph{O-variables}).
In what follows, we consider a static problem defined by
specification $spec(v)$.

A \emph{self-stabilizing protocol} (\cite{D74j}) is a protocol
that eventually reaches a \emph{legitimate configuration},
where $spec(v)$ holds at every process $v$, regardless of the initial configuration.
Once it reaches a legitimate configuration, every process never
changes its O-variables and always satisfies $spec(v)$.
 From this definition, a self-stabilizing protocol is expected to tolerate 
any number and any type of transient faults since it can eventually 
recover from any configuration affected by the transient faults.
However, the recovery from any configuration is guaranteed
only when every process correctly executes its action from 
the configuration, \emph{i.e.}, we do not consider existence of
permanently faulty processes.

\subsection{Strict stabilization}\label{sub:strict}

When (permanent) Byzantine 
processes exist, Byzantine processes may not satisfy $spec(v)$.
In addition, correct processes near the Byzantine processes
can be influenced and may be unable to satisfy $spec(v)$.
Nesterenko and Arora~\cite{NA02c} define
a \emph{strictly stabilizing protocol} as a self-stabilizing protocol 
resilient to unbounded number of Byzantine processes.

Given an integer $c$, a \emph{$c$-correct process} is a process 
 defined as follows.

\begin{definition}[$c$-correct process]
A process is $c$-correct if it is correct (\emph{i.e.} not Byzantine) and located at distance more than $c$ from any Byzantine process.
\end{definition}

\begin{definition}[$(c,f)$-containment]
\label{def:cfcontained}
A configuration $\rho$ is \emph{$(c,f)$-contained} for specification
$spec$ if, given at most $f$ Byzantine processes, in any execution
starting from $\rho$, every $c$-correct process $v$ always satisfies $spec(v)$ and never changes
its O-variables.
\end{definition}

The parameter $c$ of Definition~\ref{def:cfcontained} refers to 
the \emph{containment radius} defined in \cite{NA02c}. 
The parameter $f$ refers explicitly to the number of Byzantine processes, 
while \cite{NA02c} dealt with unbounded number of Byzantine faults 
(that is $f\in\{0\ldots n\}$).

\begin{definition}[$(c,f)$-strict stabilization]
\label{def:cfstabilizing}
A protocol is \emph{$(c,f)$-strictly stabilizing} for specification
$spec$ if, given at most $f$ Byzantine processes, any execution
$e=\rho_0,\rho_1,\ldots$ contains a configuration $\rho_i$ that
is $(c,f)$-contained for $spec$.
\end{definition}

An important limitation of the model of \cite{NA02c}
is the notion of $r$-\emph{restrictive} specifications. 
Intuitively, a specification is $r$-restrictive if it prevents 
combinations of states that belong to two processes $u$ and $v$ 
that are at least $r$ hops away. 
An important consequence related to Byzantine tolerance is that 
the containment radius of protocols solving those specifications is 
at least $r$. 
For some problems, such as the breadth-first search (BFS) spanning tree construction we consider
in this paper, $r$ can not be bounded by a constant.
In consequence, we can show that there exists no $(c,1)$-strictly stabilizing
protocol for the breadth-first search (BFS) spanning tree construction for any (finite) integer $c$.

\subsection{Strong stabilization}

To circumvent the impossibility result, \cite{MT06cb} defines a weaker notion 
than the strict stabilization. Here, the requirement to the containment 
radius is relaxed, \emph{i.e.} there may exist processes outside the 
containment radius that invalidate the specification predicate, 
due to Byzantine actions. However, the impact of Byzantine 
triggered action is limited in times: the set of Byzantine processes 
may only impact processes outside the containment radius 
a bounded number of times, even if Byzantine processes execute 
an infinite number of actions.

In the following of this section, we recall the formal definition of strong stabilization adopted in \cite{DMT10rb}.
From the states of $c$-correct processes, \emph{$c$-legitimate configurations}
and \emph{$c$-stable configurations} are defined as follows.

\begin{definition}[$c$-legitimate configuration]
A configuration $\rho$ is $c$-legitimate for \emph{spec} if every $c$-correct process $v$ satisfies $spec(v)$.
\end{definition}

\begin{definition}[$c$-stable configuration]
A configuration $\rho$ is $c$-stable if every $c$-correct process never changes the values of its O-variables as long as Byzantine processes make no action.
\end{definition}

Roughly speaking, the aim of self-stabilization is to guarantee that a distributed system
eventually reaches a $c$-legitimate and $c$-stable configuration.
However, a self-stabilizing system can be disturbed by Byzantine processes
after reaching a $c$-legitimate and $c$-stable configuration.
The \emph{$c$-disruption} represents the period where $c$-correct processes
are disturbed by Byzantine processes and is defined as follows 

\begin{definition}[$c$-disruption]
A portion of execution $e=\rho_0,\rho_1,\ldots,\rho_t$ ($t>1$) is a $c$-disruption if and only if the following holds:
\begin{enumerate}
\item $e$ is finite,
\item $e$ contains at least one action of a $c$-correct process for changing the value of an O-variable,
\item $\rho_0$ is $c$-legitimate for \emph{spec} and $c$-stable, and
\item $\rho_t$ is the first configuration after $\rho_0$ such that $\rho_t$ is $c$-legitimate for \emph{spec} and $c$-stable.
\end{enumerate}
\end{definition}

Now we can define a self-stabilizing protocol such that Byzantine processes 
may only impact processes outside the containment radius 
a bounded number of times, even if Byzantine processes execute 
an infinite number of actions.

\begin{definition}[$(t,k,c,f)$-time contained configuration]
A configuration $\rho_0$ is $(t,k,c,f)$-time contained for \emph{spec} if given at most $f$ Byzantine processes, the following properties are satisfied:
\begin{enumerate}
\item $\rho_0$ is $c$-legitimate for \emph{spec} and $c$-stable,
\item every execution starting from $\rho_0$ contains a $c$-legitimate configuration for \emph{spec} after which the values of all the O-variables of $c$-correct processes remain unchanged (even when Byzantine processes make actions repeatedly and forever), 
\item every execution starting from $\rho_0$ contains at most $t$ $c$-disruptions, and 
\item every execution starting from $\rho_0$ contains at most $k$ actions of changing the values of O-variables for each $c$-correct process.
\end{enumerate}
\end{definition}

\begin{definition}[$(t,c,f)$-strongly stabilizing protocol]
A protocol $A$ is $(t,c,f)$-strongly stabilizing if and only if starting from any arbitrary configuration, every execution involving at most $f$ Byzantine processes contains a $(t,k,c,f)$-time contained configuration that is reached after at most $l$ rounds. Parameters $l$ and $k$ are respectively the $(t,c,f)$-stabilization time and the $(t,c,f)$-process-disruption times of $A$.
\end{definition}

Note that a $(t,k,c,f)$-time contained configuration is 
a $(c,f)$-contained configuration when $t=k=0$, and thus,
$(t,k,c,f)$-time contained configuration is a generalization 
(relaxation) of a $(c,f)$-contained configuration.  
Thus, a strongly stabilizing protocol is weaker than a strictly stabilizing one
(as processes outside the containment radius may take
incorrect actions due to Byzantine influence).
However, a strongly stabilizing protocol is stronger than
a classical self-stabilizing one (that may never meet their specification
in the presence of Byzantine processes).

The parameters $t$, $k$ and $c$ are introduced to quantify the strength of
fault containment, we do not require each process to know the values of the parameters.

\section{Topology-aware Byzantine resilience}

\subsection{Topology-aware strict stabilization}

In Section \ref{sub:strict}, we saw that there exist a number of impossibility results on strict stabilization due to the notion of $r$-restrictives specifications. To circumvent this impossibility result, we describe here a weaker notion than the strict stabilization: the \emph{topology-aware strict stabilization} (denoted by TA strict stabilization for short) introduced by \cite{DMT10ra}. Here, the requirement to the containment radius is relaxed, \emph{i.e.} the set of processes which may be disturbed by Byzantine ones is not reduced to the union of $c$-neighborhood of Byzantine processes but can be defined depending on the graph topology and Byzantine processes location.

In the following, we give formal definition of this new kind of Byzantine containment. From now, $B$ denotes the set of Byzantine processes and $S_B$ (which is function of $B$) denotes a subset of $V$ (intuitively, this set gathers all processes which may be disturbed by Byzantine processes).

\begin{definition}[$S_{B}$-correct node]
A node is \emph{$S_{B}$-correct} if it is a correct node (\emph{i.e.} not Byzantine) which not belongs to $S_{B}$.
\end{definition}

\begin{definition}[$S_{B}$-legitimate configuration]
A configuration $\rho$ is \emph{$S_{B}$-legitimate} for $spec$ if every $S_{B}$-correct node $v$ is legitimate for $spec$ (\emph{i.e.} if $spec(v)$ holds).
\end{definition}

\begin{definition}[$(S_{B},f)$-topology-aware containment]
\label{def:SfTAcontained}
A configuration $\rho_{0}$ is \emph{$(S_{B},f)$-topology-aware contained} for specification $spec$ if, given at most $f$ Byzantine processes, in any execution $e=\rho_0,\rho_1,\ldots$, every configuration is $S_{B}$-legitimate and every $S_B$-correct process never changes its O-variables. 
\end{definition}

The parameter $S_{B}$ of Definition~\ref{def:SfTAcontained} refers to the \emph{containment area}. Any process which belongs to this set may be infinitely disturbed by Byzantine processes. The parameter $f$ refers explicitly to the number of Byzantine processes.

\begin{definition}[$(S_{B},f)$-topology-aware strict stabilization]
\label{def:SfTAStrictstabilizing}
A protocol is \emph{$(S_{B},f)$-topology-aware strictly stabilizing} for specification $spec$ if, given at most $f$ Byzantine processes, any execution $e=\rho_0,\rho_1,\ldots$ contains a configuration $\rho_i$ that is $(S_{B},f)$-topology-aware contained for $spec$.
\end{definition}

Note that, if $B$ denotes the set of Byzantine processes and $S_{B}=\left\{v\in V|\underset{b\in B}{min}\left(d(v,b)\right)\leq c\right\}$, then a $(S_{B},f)$-topology-aware strictly stabilizing protocol is a $(c,f)$-strictly stabilizing protocol. Then, the concept of topology-aware strict stabilization is a generalization of the strict stabilization. However, note that a TA strictly stabilizing protocol is stronger than a classical self-stabilizing protocol (that may never meet their specification in the presence of Byzantine processes).

The parameter $S_{B}$ is introduced to quantify the strength of fault containment, we do not require each process to know the actual definition of the set. Actually, the protocol proposed in this paper assumes no knowledge on the parameter.

\subsection{Topology-aware strong stabilization}

Similarly to topology-aware strict stabilization, we can weaken the notion of strong stabilization using the notion of containment area. Then, we obtain the following definition:

\begin{definition}[$S_B$-stable configuration]
A configuration $\rho$ is $S_B$-stable if every $S_B$-correct process never changes the values of its O-variables as long as Byzantine processes make no action.
\end{definition}

\begin{definition}[$S_{B}$-TA-disruption]
A portion of execution $e=\rho_0,\rho_1,\ldots,\rho_t$ ($t>1$) is a $S_{B}$-TA-disruption if and only if the followings hold:
\begin{enumerate}
\item $e$ is finite,
\item $e$ contains at least one action of a $S_{B}$-correct process for changing the value of an O-variable,
\item $\rho_0$ is $S_{B}$-legitimate for $spec$ and $S_B$-stable, and
\item $\rho_t$ is the first configuration after $\rho_0$ such that $\rho_t$ is $S_{B}$-legitimate for $spec$ and $S_B$-stable.
\end{enumerate}
\end{definition}

\begin{definition}[$(t,k,S_{B},f)$-TA time contained configuration]
A configuration $\rho_0$ is $(t,k,S_{B},$ $f)$-TA time contained for \emph{spec} if given at most $f$ Byzantine processes, the following properties are satisfied:
\begin{enumerate}
\item $\rho_0$ is $S_{B}$-legitimate for \emph{spec} and $S_B$-stable,
\item every execution starting from $\rho_0$ contains a $S_B$-legitimate configuration for \emph{spec} after which the values of all the O-variables of $S_B$-correct processes remain unchanged (even when Byzantine processes make actions repeatedly and forever), 
\item every execution starting from $\rho_0$ contains at most $t$ $S_B$-TA-disruptions, and 
\item every execution starting from $\rho_0$ contains at most $k$ actions of changing the values of O-variables for each $S_B$-correct process.
\end{enumerate}
\end{definition}

\begin{definition}[$(t,S_{B},f)$-TA strongly stabilizing protocol]
A protocol $A$ is $(t,S_{B},f)$-TA\\ strongly stabilizing if and only if starting from any arbitrary configuration, every execution involving at most $f$ Byzantine processes contains a $(t,k,S_{B},f)$-TA-time contained configuration that is reached after at most $l$ actions of each $S_{B}$-correct node. Moreover, $S_{B}$-legitimate configurations are closed by actions of $A$. Parameters $l$ and $k$ are respectively the $(t,S_{B},f)$-stabilization time and the $(t,S_{B},f)$-process-disruption time of $A$.
\end{definition}

\section{BFS Spanning Tree Construction}

In this section, we are interested in the problem of BFS spanning tree construction. That is, the system has a distinguished process called the root (and denoted by $r$) and we want to obtain a BFS spanning tree rooted to this root. We made the following hypothesis: the root $r$ is never Byzantine.

To solve this problem, each process $v$ has two O-variables: the first is $prnt_v\in N_v\cup\{\bot\}$ which is a pointer to the neighbor that is designated to be the parent of $v$ in the BFS tree and the second is $level_v\in\{0,\ldots,D\}$ which stores the depth (the number of hops from the root) of $v$ in this tree. Obviously, Byzantine process may disturb (at least) their neighbors. For example, a Byzantine process may act as the root. It is why the specification of the BFS tree construction we adopted states in fact that there exists a BFS spanning forest such that any root of this forest is either the real root of the system or a Byzantine process. More formally, we use the following specification of the problem.

\begin{definition}[BFS path]
A path $(v_0,\ldots,v_k)$ ($k\geq 1$) of $S$ is a \emph{BFS path} if and only if:
\begin{enumerate}
\item $prnt_{v_0}=\bot$, $level_{v_0}=0$, and $v_0\in B\cup\{r\}$,
\item $\forall i\in\{1,\ldots,k\}, prnt_{v_i}=v_{i-1}$ and $level_{v_i}=level_{v_{i-1}}+1$, and
\item $\forall i\in\{1,\ldots,k\}, level_{v_{i-1}}=\underset{u\in N_{v_i}}{min}\{level_{u}\}$.
\end{enumerate}
\end{definition}

We define the specification predicate $spec(v)$ of the BFS spanning tree construction as follows.
\[spec(v) : \begin{cases}
 prnt_v = \bot \text{ and } level_v = 0 \text{ if } v \text{ is the root } r \\
 \text{there exists a BFS path } (v_0,\ldots,v_k) \text{ such that } v_k=v \text{ otherwise}
\end{cases}\]

In the case where any process is correct, note that $spec$ implies the existence of a BFS spanning tree rooted to the real root. The well-known $min+1$ protocol solves this problem in a self-stabilizing way (see \cite{HC92j}). In the following of this section, we assume that some processes may be Byzantine and we study the Byzantine containment properties of this protocol. We show that this self-stabilizing protocol has moreover optimal Byzantine containment properties.

In more details, we prove first that there exists neither strictly nor strongly stabilizing solution to the BFS spanning tree construction (see Theorems \ref{th:impStricte} and \ref{th:impStrong}). Then, we demonstrate in Theorems \ref{th:possTAStricte} and \ref{th:possTAStrong} that the $min+1$ protocol is both $(S_{B},f)$-TA strictly and $(t,S_{B}^*,f)$-TA strongly stabilizing where $f\leq n-1$, $t=2m$, and 
\[\begin{array}{ccc}
S_{B}&=&\left\{v\in V\left|\underset{b\in B}{min}\left(d(v,b)\right)\leq d(r,v)\right.\right\}\\
S_{B}^*&=&\left\{v\in V\left|\underset{b\in B}{min}\left(d(v,b)\right)<d(r,v)\right.\right\}
\end{array}\]
Figure \ref{fig:ExBFS} provides an example of these containment areas. Finally, we show that these containment areas are in fact optimal (see Theorem \ref{th:impTAStricte} and \ref{th:impTAStrong}).

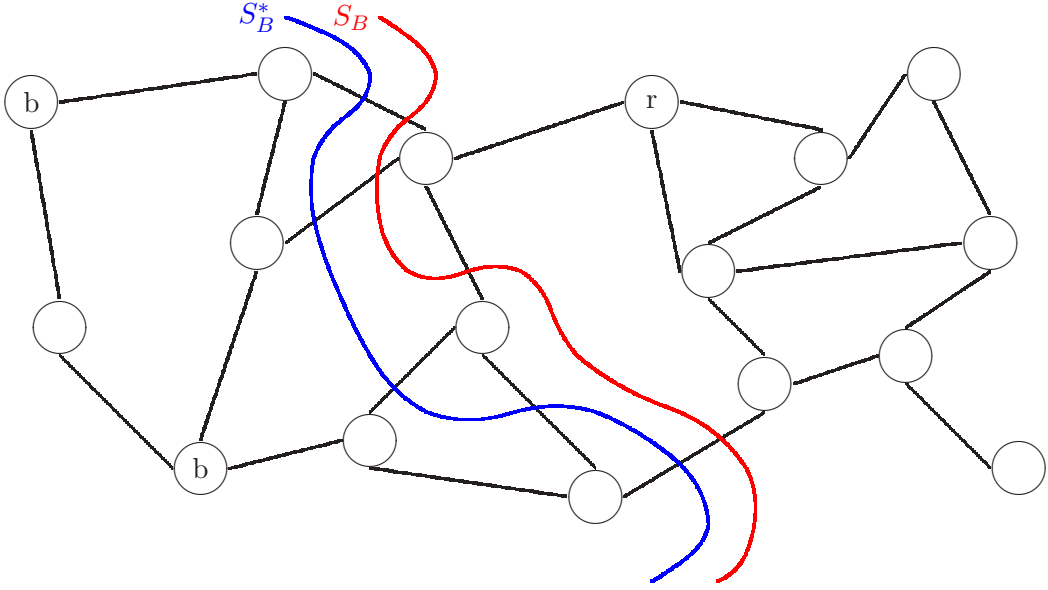
\begin{figure}[t]
\noindent \begin{centering} \ifx\JPicScale\undefined\def\JPicScale{0.75}\fi
\unitlength \JPicScale mm
\begin{picture}(185,100)(0,0)
\linethickness{0.3mm}
\put(5,85){\circle{10}}

\linethickness{0.3mm}
\put(50,90){\circle{10}}

\linethickness{0.3mm}
\put(10,45){\circle{10}}

\linethickness{0.3mm}
\put(75,75){\circle{10}}

\linethickness{0.3mm}
\put(115,85){\circle{10}}

\linethickness{0.3mm}
\put(145,75){\circle{10}}

\linethickness{0.3mm}
\put(125,55){\circle{10}}

\linethickness{0.3mm}
\put(135,35){\circle{10}}

\linethickness{0.3mm}
\put(105,15){\circle{10}}

\linethickness{0.3mm}
\put(85,45){\circle{10}}

\linethickness{0.3mm}
\put(65,25){\circle{10}}

\linethickness{0.3mm}
\put(35,20){\circle{10}}

\linethickness{0.3mm}
\put(45,60){\circle{10}}

\linethickness{0.3mm}
\multiput(10,85)(0.83,0.12){42}{\line(1,0){0.83}}
\linethickness{0.3mm}
\multiput(55,90)(0.24,-0.12){83}{\line(1,0){0.24}}
\linethickness{0.3mm}
\multiput(80,75)(0.36,0.12){83}{\line(1,0){0.36}}
\linethickness{0.3mm}
\multiput(120,85)(0.6,-0.12){42}{\line(1,0){0.6}}
\linethickness{0.3mm}
\multiput(125,60)(0.24,0.12){83}{\line(1,0){0.24}}
\linethickness{0.3mm}
\multiput(115,80)(0.12,-0.6){42}{\line(0,-1){0.6}}
\linethickness{0.3mm}
\multiput(125,50)(0.12,-0.12){83}{\line(1,0){0.12}}
\linethickness{0.3mm}
\multiput(110,15)(0.2,0.12){125}{\line(1,0){0.2}}
\linethickness{0.3mm}
\multiput(85,40)(0.12,-0.12){167}{\line(1,0){0.12}}
\linethickness{0.3mm}
\multiput(65,20)(0.83,-0.12){42}{\line(1,0){0.83}}
\linethickness{0.3mm}
\multiput(65,30)(0.12,0.12){125}{\line(1,0){0.12}}
\linethickness{0.3mm}
\multiput(75,70)(0.12,-0.24){83}{\line(0,-1){0.24}}
\linethickness{0.3mm}
\multiput(50,60)(0.16,0.12){125}{\line(1,0){0.16}}
\linethickness{0.3mm}
\multiput(45,65)(0.12,0.48){42}{\line(0,1){0.48}}
\linethickness{0.3mm}
\multiput(5,80)(0.12,-0.71){42}{\line(0,-1){0.71}}
\linethickness{0.3mm}
\multiput(10,40)(0.12,-0.12){167}{\line(1,0){0.12}}
\linethickness{0.3mm}
\multiput(40,20)(0.48,0.12){42}{\line(1,0){0.48}}
\linethickness{0.3mm}
\multiput(35,25)(0.12,0.36){83}{\line(0,1){0.36}}
\put(5,85){\makebox(0,0)[cc]{b}}

\put(115,85){\makebox(0,0)[cc]{r}}

\put(35,20){\makebox(0,0)[cc]{b}}

\linethickness{0.3mm}
\put(165,90){\circle{10}}

\linethickness{0.3mm}
\put(175,60){\circle{10}}

\linethickness{0.3mm}
\put(160,40){\circle{10}}

\linethickness{0.3mm}
\put(180,20){\circle{10}}

\linethickness{0.3mm}
\multiput(150,75)(0.12,0.18){83}{\line(0,1){0.18}}
\linethickness{0.3mm}
\multiput(165,85)(0.12,-0.24){83}{\line(0,-1){0.24}}
\linethickness{0.3mm}
\multiput(160,45)(0.18,0.12){83}{\line(1,0){0.18}}
\linethickness{0.3mm}
\multiput(140,35)(0.36,0.12){42}{\line(1,0){0.36}}
\linethickness{0.3mm}
\multiput(160,35)(0.12,-0.12){125}{\line(1,0){0.12}}
\linethickness{0.3mm}
\multiput(130,55)(0.95,0.12){42}{\line(1,0){0.95}}

\textcolor{blue}{
\linethickness{0.3mm}
\qbezier(50,100)(50.83,99.93)(57.06,97.21)
\qbezier(57.06,97.21)(63.29,94.49)(65,90)
\qbezier(65,90)(65.26,85.76)(60.89,82.56)
\qbezier(60.89,82.56)(56.52,79.36)(55,75)
\qbezier(55,75)(53.87,68.57)(55.52,62.29)
\qbezier(55.52,62.29)(57.18,56.02)(60,50)
\qbezier(60,50)(62.53,43.95)(65.96,38.48)
\qbezier(65.96,38.48)(69.38,33.02)(75,30)
\qbezier(75,30)(81.99,27.36)(89.88,29.76)
\qbezier(89.88,29.76)(97.76,32.15)(105,30)
\qbezier(105,30)(112.12,27.2)(118.35,22.15)
\qbezier(118.35,22.15)(124.59,17.09)(125,10)
\qbezier(125,10)(124.38,6.28)(120,3.23)
\qbezier(120,3.23)(115.62,0.18)(115,0)
\put(45,100){\makebox(0,0)[cc]{$S_B^*$}}
}

\textcolor{red}{
\linethickness{0.3mm}
\qbezier(65,100)(65.58,99.82)(69.82,96.77)
\qbezier(69.82,96.77)(74.07,93.73)(75,90)
\qbezier(75,90)(74.87,85.67)(70.62,82.52)
\qbezier(70.62,82.52)(66.37,79.36)(65,75)
\qbezier(65,75)(64.01,69.59)(65,64)
\qbezier(65,64)(65.98,58.4)(70,55)
\qbezier(70,55)(74.17,52.58)(79.74,54.65)
\qbezier(79.74,54.65)(85.3,56.72)(90,55)
\qbezier(90,55)(93.77,52.57)(95.38,48.06)
\qbezier(95.38,48.06)(97,43.55)(100,40)
\qbezier(100,40)(107.09,34)(116.12,30.77)
\qbezier(116.12,30.77)(125.14,27.54)(130,20)
\qbezier(130,20)(131.77,16.55)(131.64,12.59)
\qbezier(131.64,12.59)(131.52,8.63)(130,5)
\qbezier(130,5)(129.29,3.31)(127.99,2.01)
\qbezier(127.99,2.01)(126.69,0.71)(125,0)
\put(60,100){\makebox(0,0)[cc]{$S_B$}}
}

\end{picture}
  \par\end{centering}
 \caption{Example of containment areas for BFS spanning tree construction.}
\label{fig:ExBFS}
\end{figure}

\subsection{Impossibility results}

\begin{theorem}\label{th:impStricte}
Even under the central daemon, there exists no $(c,1)$-strictly stabilizing protocol for BFS spanning tree construction where $c$ is any (finite) integer.
\end{theorem}

\begin{proof}
This result is a direct application of Theorem 4 of \cite{NA02c} (note that the specification of BFS tree construction is $D$-restrictive in the worst case where $D$ is the diameter of the system).
\end{proof}

\begin{theorem}\label{th:impStrong}
Even under the central daemon, there exists no $(t,c,1)$-strongly stabilizing protocol for BFS spanning tree construction where $t$ and $c$ are any (finite) integers.
\end{theorem}

\begin{proof}
Let $t$ and $c$ be (finite) integers. Assume that there exists a $(t,c,1)$-strongly stabilizing protocol $\mathcal{P}$ for BFS spanning tree construction under the central daemon. Let $S=(V,E)$ be the following system $V=\{p_0=r,p_1,\ldots,p_{2c+2},p_{2c+3}=b\}$ and $E=\{\{p_i,p_{i+1}\},i\in\{0,\ldots,2c+2\}\}$. Process $p_0$ is the real root and process $b$ is a Byzantine one.

Assume that the initial configuration $\rho_0$ of $S$ satisfies: $level_r=level_b=0$, $prnt_r=prnt_b=\bot$ and other variables of $b$ (if any) are identical to those of $r$ (see Figure \ref{fig:impBFS}). Assume now that $b$ takes exactly the same actions as $r$ (if any) immediately after $r$ (note that $d(r,b)>c$ and hence $level_r=0$ and $prnt_r=\bot$ still hold by closure and then $level_b=0$ and $prnt_b=\bot$ still hold too). Then, by symmetry of the execution and by convergence of $\mathcal{P}$ to $spec$, we can deduce that the system reaches in a finite time a configuration $\rho_1$ (see Figure \ref{fig:impBFS}) in which: $\forall i\in\{1,\ldots,c+1\},level_{p_i}=i$ and $prnt_{p_i}=p_{i-1}$ and  $\forall i\in\{c+2,\ldots,2c+2\},level_{p_i}=2c+3-i$ and $prnt_{p_i}=p_{i+1}$ (because this configuration is the only one in which all correct process $v$ such that $d(v,b)>c$ satisfies $spec(v)$ when $level_r=level_b=0$ and $prnt_r=prnt_b=\bot$). Note that $\rho_1$ is $0$-legitimate and $0$-stable and \emph{a fortiori} $c$-legitimate and $c$-stable.

Assume now that the Byzantine process acts as a correct process and executes correctly $\mathcal{P}$. Then, by convergence of $\mathcal{P}$ in fault-free systems (remember that a $(t,c,1)$-strongly stabilizing protocol is a special case of self-stabilizing protocol), we can deduce that the system reaches in a finite time a configuration $\rho_2$ (see Figure \ref{fig:impBFS}) in which: $\forall i\in\{1,\ldots,2c+3\},level_{p_i}=i$ and $prnt_{p_i}=p_{i-1}$ (because this configuration is the only one in which every process $v$ satisfies $spec(v)$). Note that the portion of execution between $\rho_1$ and $\rho_2$ contains at least one $c$-perturbation ($p_{c+2}$ is a $c$-correct process and modifies at least once its O-variables) and that $\rho_2$ is $0$-legitimate and $0$-stable and \emph{a fortiori} $c$-legitimate and $c$-stable.

Assume now that the Byzantine process $b$ takes the following state: $level_b=0$ and $prnt_b=\bot$. This step brings the system into configuration $\rho_3$ (see Figure \ref{fig:impBFS}). From this configuration, we can repeat the execution we constructed from $\rho_0$. By the same token, we obtain an execution of $\mathcal{P}$ which contains $c$-legitimate and $c$-stable configurations (see $\rho_1$) and an infinite number of $c$-perturbation which contradicts the $(t,c,1)$-strong stabilization of $\mathcal{P}$.
\end{proof}

\begin{figure}[t]
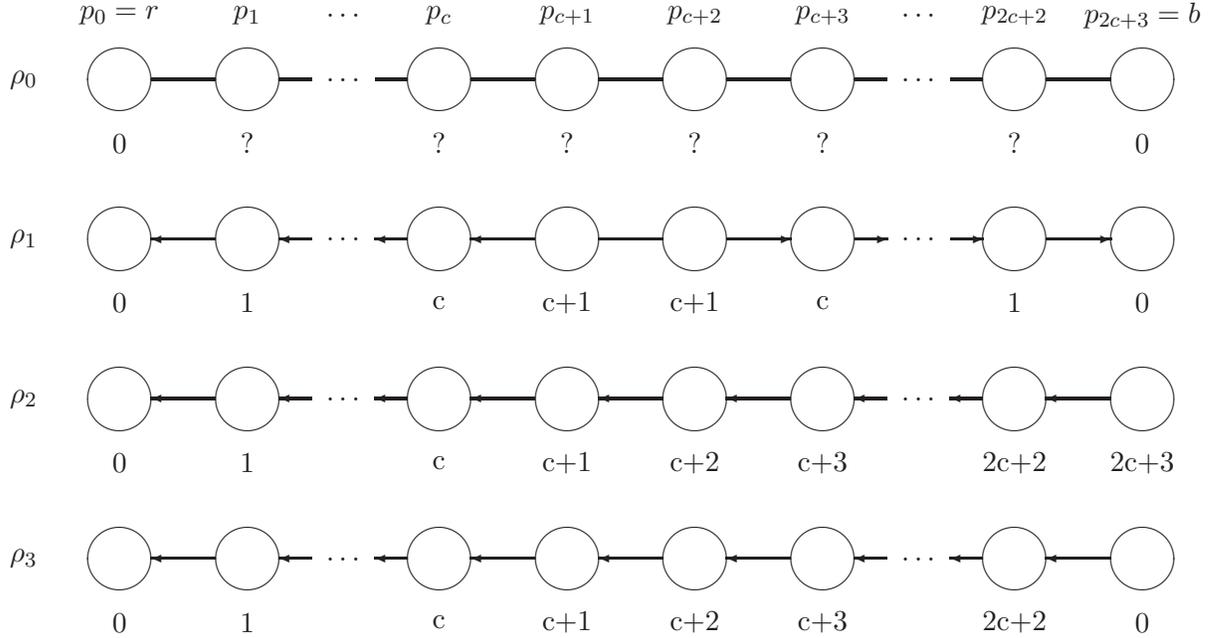

\noindent \begin{centering} \include{impBFS}
  \par\end{centering}
 \caption{Configurations used in proof of Theorem \ref{th:impStrong}.}
\label{fig:impBFS}
\end{figure}

\subsection{Byzantine containment properties of the $min+1$ protocol}

In the $min+1$ protocol, as in many self-stabilizing tree construction protocols, each process $v$ checks locally the consistence of its $level_v$ variable with respect to the one of its neighbors. When it detects an inconsistency, it changes its $prnt_v$ variable in order to choose a ``better'' neighbor. The notion of ``better'' neighbor is based on the global desired property on the tree (here, the BFS requirement implies to choose one neighbor with the minimum level).

When the system may contain Byzantine processes, they may disturb their neighbors by providing alternatively ``better'' and ``worse'' states. 

The $min+1$ protocol chooses an arbitrary one of the ``better'' neighbors (that is, neighbors with the minimal level).  Actually this strategy allows us to achieve the $(S_B,f)$-TA strict stabilization but is not sufficient to achieve the $(t,S^*_B,f)$-TA strong stabilization. To achieve the $(t,S^*_B,f)$-TA strong stabilization, we must bring a slight modification to the protocol: we choose a ``better'' neighbor with a round robin order (among its neighbors with the minimal level). 

Algorithm \ref{algo:BFS} presents our BFS spanning tree construction protocol $\mathcal{SSBFS}$ which is both $(S_{B},f)$-TA strictly and $(t,S_{B}^*,f)$-TA strongly stabilizing (where $f\leq n-1$ and $t=2m$) provided that the root is never Byzantine.

\begin{algorithm}
\caption{$\mathcal{SSBFS}$: A TA strictly and TA strongly stabilizing protocol for BFS tree construction}\label{algo:BFS}
\begin{description}
\item{Data:}~\\
$N_v$: totally ordered set of neighbors of $v$ 
\item{Variables:}~\\
$prnt_v\in N_v\cup\{\bot\}$: pointer on the parent of $v$ in the tree.\\
$level_v\in\mathbb{N}$: integer
\item{Macro:}~\\
For any subset $A\subseteq N_v$, $choose(A)$ returns the first element of $A$ which is bigger than $prnt_v$ (in a round-robin fashion).
\item{Rules:}~\\
$\boldsymbol{(R_r)}::(v=r)\wedge((prnt_v\neq\bot)\vee(level_v\neq 0))\longrightarrow prnt_v:=\bot;level_v:=0$\\
$\boldsymbol{(R_v)}::(v\neq r)\wedge\left((prnt_v=\bot)\vee(level_v\neq level_{prnt_v}+1)\vee(level_{prnt_v}\neq\underset{q\in N_v}{min}\{level_q\})\right)\longrightarrow$\\$~~~~~~~~~~~~~~~ prnt_v:=choose\left(\left\{p\in N_v\left|level_p=\underset{q\in N_v}{min}\{level_q\}\right.\right\}\right);level_v:=level_{prnt_v}+1$
\end{description}
\end{algorithm}

In the following of this section, we provide proofs of topology-aware strict and strong stabilization of $\mathcal{SSBFS}$. First at all, remember that the real root $r$ can not be a Byzantine process by hypothesis. Note that the subsystems whose set of nodes are respectively $V\setminus S_B$ and $V\setminus S_B^*$ are connected by construction.

\paragraph{$(S_B,n-1)$-TA strict stabilization}~\\

Given a configuration $\rho\in C$ and an integer $d\in\{0,\ldots,D\}$, let us define the following predicate: 
\[I_d(\rho)\equiv \forall v\in V,level_v\geq min\left\{d,\underset{u\in B\cup\{r\}}{min}\{d(v,u)\}\right\}\]

\begin{lemma}\label{lem:Iclosed}
For any integer $d\in\{0,\ldots,D\}$, the predicate $I_d$ is closed.
\end{lemma}

\begin{proof}
Let $d$ be an integer such that $d\in\{0,\ldots,D\}$. Let $\rho\in C$ be a configuration such that $I_d(\rho)=true$ and $\rho'\in C$ be a configuration such that $\rho \stackrel{R}{\mapsto} \rho'$ is a step of $\mathcal{SSBFS}$.

If the root process $r\in R$ (respectively a Byzantine process $b\in R$), then we have $level_r=0$ (respectively $level_b\geq 0$) in $\rho'$ by construction of $\boldsymbol{(R_r)}$ (respectively by definition of $level_b$). Hence, $level_r\geq min\left\{d, \underset{u\in B\cup\{r\}}{min}\{d(r,u)\}\right\}=0$ (respectively $level_b\geq min\left\{d,\underset{u\in B\cup\{r\}}{min}\{d(b,u)\}\right\}=0$).

If a correct process $v\in R$ satisfies $v\neq r$, then there exists a neighbor $p$ of $v$ which satisfies the following property in $\rho$ (since $v$ is activated and $I_d(\rho)=true$):
\[level_p=\underset{q\in N_v}{min}\{level_q\}\geq min\left\{d,\underset{u\in B\cup\{r\}}{min}\{d(p,u)\}\right\}\]
Once, $v$ is activated, we have: $level_v=level_p+1$ in $\rho'$. Let be $\delta=\underset{u\in B\cup\{r\}}{min}\{d(v,u)\}$. Then, we have: $\underset{u\in B\cup\{r\}}{min}\{d(p,u)\}\geq \delta-1$ (otherwise, we have a contradiction with the fact that $\delta=\underset{u\in B\cup\{r\}}{min}\{d(v,u)\}$ and that $v$ and $p$ are neighbors). Consequently, $\rho'$ satisfies: 
\[\begin{array}{rcl}
level_v=level_p+1&\geq&min\left\{d,\underset{u\in B\cup\{r\}}{min}\{d(p,u)\}\right\}+1\\
&\geq& min\{d,\delta-1\}+1\\
&\geq& min\{d,\delta\}\\
&\geq&min\left\{d,\underset{u\in B\cup\{r\}}{min}\{d(v,u)\}\right\}
\end{array}\]

We can deduce that $I_d(\rho')=true$, that concludes the proof.
\end{proof}

Let $\mathcal{LC}$ be the following set of configurations:
\[\mathcal{LC}=\left\{\rho\in C\left|(\rho \text{ is }S_B\text{-legitimate for }spec)\wedge(I_D(\rho)=true)\right.\right\}\]

\begin{lemma}\label{lem:SBTAcontained}
Any configuration of $\mathcal{LC}$ is $(S_B,n-1)$-TA contained for $spec$.
\end{lemma}

\begin{proof}
Let $\rho$ be a configuration of $\mathcal{LC}$. By construction, $\rho$ is $S_B$-legitimate for $spec$. 

In particular, the root process satisfies: $prnt_r=\bot$ and $level_r=0$. By construction of $\mathcal{SSBFS}$, $r$ is not enabled and then never modifies its O-variables (since the guard of the rule of $r$ does not involve the state of its neighbors). 

In the same way, any process $v\in V\setminus (S_B\cup\{r\})$ satisfies: $prnt_v\in N_v$, $level_v=level_{prnt_v}+1$, and $level_{prnt_v}=\underset{u\in N_v}{min}\{level_u\}$. Note that, as $v\in V\setminus (S_B\cup\{r\})$ and $spec(v)$ holds in $\rho$, we have: $level_v=d(v,r)$. Hence, process $v$ is not enabled in $\rho$. It remains so until none of its neighbors $u$ modifies its $level_u$ variable to a value $\alpha$ such that $\alpha\leq level_v-2$.

Assume that there exists an execution $e$ starting from $\rho$ in which a neighbor $u$ of a process $v\in V\setminus (S_B\cup\{r\})$ modifies $level_u$ to satisfy $level_u\leq level_v-2$ (without loss of generality, assume that $u$ is the first process to modify $level_u$ in such a way in $e$). Note that  $\underset{p\in B\cup\{r\}}{min}\{d(u,p)\}\geq d(v,r)-1$ (otherwise, we have a contradiction with the fact that $d(v,r)=\underset{p\in B\cup\{r\}}{min}\{d(v,p)\}$ and that $v$ and $u$ are neighbors). Hence, we have:
\[\begin{array}{rcl}
\underset{p\in B\cup\{r\}}{min}\{d(u,p)\}&\geq& d(v,r)-1\\
&>&d(v,r)-2\\
&>&level_u
\end{array}\]
This contradicts the closure of predicate $I_D$ established in Lemma \ref{lem:Iclosed}.

Consequently, there exists no such execution and process $v$ remains infinitely disabled and then never modifies its O-variables. This concludes the proof.
\end{proof}

\begin{lemma}\label{lem:convergenceLC}
Starting from any configuration, any execution of $\mathcal{SSBFS}$ reaches a configuration of $\mathcal{LC}$ in a finite time.
\end{lemma}

\begin{proof}
We are going to prove the following property by induction on $d\in\{0,\ldots ,D\}$:

$(\mathcal{P}_d)$: Starting from any configuration, any run of $\mathcal{SSBFS}$ reaches a configuration $\rho$ such that $I_d(\rho)=true$ and in which any process $v\notin S_B$ such that $d(v,r)\leq d$ satisfies $spec(v)$.


\begin{description}
\item[Initialization:] $d=0$.\\
Let $\rho$ be an arbitrary configuration. Then, it is obvious that $I_0(\rho)$ is satisfied.

If a process $v\notin S_B$ satisfies $d(v,r)\leq 0$, then $v=r$. If $v$ does not satisfy $spec(v)$ in $\rho$, then $v$ is continuously enabled. Since the scheduling is fair, $v$ is activated in a finite time and then $v$ satisfies $spec(v)$ in a finite time. Then, we proved that $(\mathcal{P}_0)$ holds.

\item[Induction:] $d\geq 1$ and $\mathcal{P}_{d-1}$ is true.\\
We know, by $\mathcal{P}_{d-1}$, that any run of $\mathcal{SSBFS}$ under a distributed fair scheduler reaches a configuration $\rho$ such that $I_{d-1}(\rho)=true$ and in which any process $v\notin S_B$ such that $d(v,r)\leq d-1$ satisfies $spec(v)$.

Let $E_d=\left\{v\in V\left|\underset{u\in B\cup\{r\}}{min}\{d(v,u)\}\geq d\right.\right\}$. Note that $I_{d-1}(\rho)$ implies that $\forall v\in E_d, level_v\geq d-1$ (since $\forall v\in E_d, min\left\{d-1,\underset{u\in B\cup\{r\}}{min}\{d(v,u)\}\right\}=d-1$ by construction).

Note that any process $v\in E_d$ such that $level_v=d-1$ is enabled by $\boldsymbol{(R_v)}$ since we have: $level_{prnt_v}\geq d-1$ (by $I_{d-1}(\rho)$ and the fact that $prnt_v$ is a neighbor of $v$) and thus $level_v=d-1<level_{prnt_v}+1$. Moreover, this rule remains enabled until $v$ is activated by closure of $I_{d-1}(\rho)$ (see Lemma \ref{lem:Iclosed}). As the scheduling is fair, we deduce that any process $v\in E_d$ such that $level_v=d-1$ is activated in any run starting from $\rho$ and $level_v \ge d$ holds. Then, we can conclude that any run starting from $\rho$ reaches in a finite time a configuration $\rho'$ such that $I_d(\rho')=true$.

Let $v\notin S_B$ be a process such that $d(r,v)=d$. We distinguish the following two cases:

\begin{description}
\item[Case 1:] $spec(v)$ holds in $\rho'$ (and then $level_v=d$).\\
By closure of $I_d$, any configuration of any run starting from $\rho'$ satisfies $I_d$. Moreover, $v$ satisfies $d(v,r)<\underset{u\in B}{min}\{d(v,u)\}$. Hence, there exists a BFS path from $v$ to $r$. By construction, process $v$ is then not enabled (remind that any neighbor $u$ of $v$ satisfies: $level_u\geq min\left\{d,\underset{w\in B\cup\{r\}}{min}\{d(u,w)\}\right\}\geq d$). In conclusion, $v$ always satisfies $spec(v)$ in any run starting from $\rho'$.
\item[Case 2:] $spec(v)$ does not hold in $\rho'$.\\
By construction of $\rho'$, we can split $N_v$ into two sets $S$ and $\bar{S}$ such that any process $u$ of $S$ satisfies $level_u=d(r,u)=d-1$ and $spec(u)$ (and thus there exists a BFS path from $u$ to $r$) and any process $u$ of $\bar{S}$ satisfies $level_u\geq d$ (remind that $I_d(\rho')=true$ and then $level_u\geq min\left\{d,\underset{p\in B\cup\{r\}}{min}\{d(u,p)\}\right\}\geq d$).

As $spec(v)$ does not hold in $\rho'$, we can deduce that $v$ is enabled in $\rho'$. As $I_d$ is closed (by Lemma \ref{lem:Iclosed}), we can deduce that $v$ remains enabled. Since the scheduling is fair, we conclude that $v$ is activated in a finite time in any run starting from $\rho'$ and then $prnt_v$ is a process of $S$ that implies that $v$ satisfies $spec(v)$ in a finite time in any run starting from $\rho'$.
\end{description}
In conclusion, $\mathcal{P}_{d}$ is true, that ends the induction.
\end{description}

Then, it is easy to see that $\mathcal{P}_D$ (where $D$ is the diameter of the system) implies the result.
\end{proof}

\begin{theorem}\label{th:possTAStricte}
$\mathcal{SSBFS}$ is a $(S_B,n-1)$-TA strictly stabilizing protocol for $spec$.
\end{theorem}

\begin{proof}
This result is a direct consequence of Lemmas \ref{lem:SBTAcontained} and \ref{lem:convergenceLC}.
\end{proof}

\paragraph{$(2m,S_B^*,n-1)$-TA strong stabilization}

Let be $E_B=S_B\setminus S_B^*$ (\emph{i.e.} $E_B$ is the set of process $v$ such that $d(r,v)=\underset{b\in B}{min}\{d(v,b)\}$).

\begin{lemma}\label{lem:LCdegvactions}
If $\rho$ is a configuration of $\mathcal{LC}$, then any process $v\in E_B$ is activated at most $\Delta_v$ times in any execution starting from $\rho$.
\end{lemma}

\begin{proof}
Let $\rho$ be a configuration of $\mathcal{LC}$ and $v$ a process of $E_B$. By construction, there exists a neighbor $u$ of $v$ such that $u\in V\setminus S_B$. Then, we know that $spec(u)$ holds in $\rho$. By Lemma \ref{lem:SBTAcontained}, we are ensured that $spec(u)$ remains true in any configuration of any execution starting from $\rho$. In particular, $level_u=d(r,u)$. By closure of $I_D(\rho)$, we know that $level_p\geq d(r,u)$ for any neighbor $p$ of $v$. Consequently, $level_u=\underset{q\in N_v}{min}\{level_q\}$. This implies that, if $prnt_v=u$ and $level_v=level_u+1$ in a configuration $\rho'$, then $spec(v)$ is satisfied and $v$ takes no actions in any execution starting from $\rho'$.

Then, the construction of the macro $choose$ implies that $u$ is chosen as $v$'s parent in at most $\Delta_v$ actions of $v$. This implies the result.
\end{proof}

\begin{lemma}\label{lem:activatedorspec}
If $\rho$ is a configuration of $\mathcal{LC}$ and $v$ is a process such that $v\in E_B$, then for any execution $e$ starting from $\rho$ either
\begin{enumerate}
\item there exists a configuration $\rho'$ of $e$ such that $spec(v)$ is always satisfied after $\rho'$, or
\item $v$ is activated in $e$.
\end{enumerate}
\end{lemma}

\begin{proof}
Let $\rho$ be a configuration of $\mathcal{LC}$ and $v$ be a process such that $v\in E_B$. By contradiction, assume that there exists an execution starting from $\rho$ such that $(i)$ $spec(v)$ is infinitely often false in $e$ and $(ii)$ $v$ is never activated in $e$.

For any configuration $\rho$, let us denote by $P_v(\rho)=(v_0=v,v_1=prnt_v,v_2=prnt_{v_1},\ldots,v_k=prnt_{v_{k-1}},p_v=prnt_{v_k})$ the maximal sequence of processes following pointers $prnt$ (maximal means here that either $prnt_{p_v}=\bot$ or $p_v$ is the first process such that there $p_v=v_i$ for some $i\in\{0,\ldots,k\}$).

Let us study the following cases:
\begin{description}
\item[Case 1:] $prnt_v\in V\setminus S_B$ in $\rho$.\\
Since $\rho\in\mathcal{LC}$, $prnt_v$ satisfies $spec(prnt_v)$ in $\rho$ and in any execution starting from $\rho$ (by Lemma \ref{lem:SBTAcontained}). If $v$ does not satisfy $spec(v)$ in $\rho$, then we have $level_v\neq level_{prnt_v}+1$ in $\rho$. Then, $v$ is continuously enabled in $e$ and we have a contradiction between assumption $(ii)$ and the fairness of the scheduling. This implies that $v$ satisfies $spec(v)$ in $\rho$. The closure of $I_D$ (established in Lemma \ref{lem:Iclosed}) ensures us that $v$ is never enabled in any execution starting from $\rho$. Hence, $spec(v)$ remains true in any execution starting from $\rho$. This contradicts the assumption $(i)$ on $e$.
\item[Case 2:] $prnt_v\notin V\setminus S_B$ in $\rho$.\\
By the assumption $(i)$ on $e$, we can deduce that there exists infinitely many configurations $\rho'$ such that a process of $P_v(\rho')$ is enabled. By construction, the length of $P_v(\rho')$ is finite for any configuration $\rho'$ and there exists only a finite number of processes in the system. Consequently, there exists at least one process which is infinitely often enabled in $e$. Since the scheduler is fair, we can conclude that there exists at least one process which is infinitely often activated in $e$.

Let $A_e$ be the set of processes which are infinitely often activated in $e$. Note that $v\notin A_e$ by assumption $(ii)$ on $e$. Let $e'=\rho'\ldots$ be the suffix of $e$ which contains only activations of processes of $A_e$. Let $p$ be the first process of $P_v(\rho')$ which belongs to $A_e$ ($p$ exists since at least one process of $P_v$ is enabled when $spec(v)$ is false). By construction, the prefix of $P_v(\rho'')$ from $v$ to $p$ in any configuration $\rho''$ of $e$ remains the same as the one of $P_v(\rho')$. Let $p'$ be the process such that $prnt_{p'}=p$ in $e'$ ($p'$ exists since $v\neq p$ implies that the prefix of $P_v(\rho')$ from $v$ to $p$ counts at least two processes). As $p$ is infinitely often activated and as any activation of $p$ modifies the value of $level_p$ (it takes at least two different values in $e'$), we can deduce that $p'$ is infinitely often enabled in $e'$ (since the value of $level_{p'}$ is constant by construction of $e'$ and $p$). Since the scheduler is fair, $p'$ is activated in a finite time in $e'$, that contradicts the construction of $p$. 
\end{description}
In the two cases, we obtain a contradiction with the construction of $e$, that proves the result.
\end{proof}

Let $\mathcal{LC^*}$ be the following set of configurations:
\[\mathcal{LC^*}=\left\{\rho\in C\left|(\rho \text{ is }S_B^*\text{-legitimate for }spec)\wedge(I_D(\rho)=true)\right.\right\}\]

Note that, as $S_B^*\subseteq S_B$, we can deduce that $\mathcal{LC^*}\subseteq\mathcal{LC}$. Hence, properties of Lemmas \ref{lem:LCdegvactions} and \ref{lem:activatedorspec} also apply to configurations of $\mathcal{LC^*}$.

\begin{lemma}\label{lem:SB*TAtimecontained}
Any configuration of $\mathcal{LC^*}$ is $(2m,\Delta,S_B^*,n-1)$-TA time contained for $spec$.
\end{lemma}

\begin{proof}
Let $\rho$ be a configuration of $\mathcal{LC^*}$. As $S_B^*\subseteq S_B$, we know by Lemma \ref{lem:SBTAcontained} that any process $v$ of $V\setminus S_B$ satisfies $spec(v)$ and takes no action in any execution starting from $\rho$.

Let $v$ be a process of $E_B$. By Lemmas \ref{lem:LCdegvactions} and \ref{lem:activatedorspec}, we know that $v$ takes at most $\Delta_v$ actions in any execution starting from $\rho$. Moreover, we know that $v$ satisfies $spec(v)$ after its last action (otherwise, we obtain a contradiction between the two lemmas). Hence, any process of $E_B$ takes at most $\Delta_v\leq \Delta$ actions and then, there are at most $\underset{v\in V}{\sum}\Delta_v=2m$ $S_B^*$-TA-disruptions in any execution starting from $\rho$.

By definition of a TA time contained configuration, we obtain the result.
\end{proof}

\begin{lemma}\label{lem:convergenceLC*}
Starting from any configuration, any execution of $\mathcal{SSBFS}$ reaches a configuration of $\mathcal{LC^*}$ in a finite time under a distributed fair scheduler.
\end{lemma}

\begin{proof}
Let $\rho$ be an arbitrary configuration. We know by Lemma \ref{lem:convergenceLC} that any execution starting from $\rho$ reaches in a finite time a configuration $\rho'$ of $\mathcal{LC}$. 

Let $v$ be a process of $E_B$. By Lemmas \ref{lem:LCdegvactions} and \ref{lem:activatedorspec}, we know that $v$ takes at most $\Delta_v$ actions in any execution starting from $\rho'$. Moreover, we know that $v$ satisfies $spec(v)$ after its last action (otherwise, we obtain a contradiction between the two lemmas). This implies that any execution starting from $\rho'$ reaches a configuration $\rho''$ such that any process $v$ of $E_B$ satisfies $spec(v)$. It is easy to see that $\rho''\in\mathcal{LC^*}$, that ends the proof.
\end{proof}

\begin{theorem}\label{th:possTAStrong}
$\mathcal{SSBFS}$ is a $(2m,S_B^*,n-1)$-TA strongly stabilizing protocol for $spec$.
\end{theorem}

\begin{proof}
This result is a direct consequence of Lemmas \ref{lem:SB*TAtimecontained} and \ref{lem:convergenceLC*}.
\end{proof}

\subsection{Optimality of containment areas of the $min+1$ protocol}

\begin{theorem}\label{th:impTAStricte}
Even under the central daemon, there exists no $(A_B,1)$-TA strictly stabilizing protocol for BFS spanning tree construction where $A_B\varsubsetneq S_B$.
\end{theorem}

\begin{proof}
This is a direct application of the Theorem 2 of \cite{DMT10ra}.
\end{proof}

\begin{theorem}\label{th:impTAStrong}
Even under the central daemon, there exists no $(t,A_B,1)$-TA strongly stabilizing protocol for BFS spanning tree construction where $A_B\varsubsetneq S_B$ and $t$ is any (finite) integer.
\end{theorem}

\begin{proof}
Let $\mathcal{P}$ be a $(t,A_B,1)$-TA strongly stabilizing protocol for BFS spanning tree construction protocol where $A_B\varsubsetneq S_B^*$ and $t$ is a finite integer. We must distinguish the following cases:

Consider the following system: $V=\{r,u,u',v,v',b\}$ and $E=\{\{r,u\},\{r,u'\},$ $\{u,v\},\{u',v'\},$ $\{v,b\},\{v',b\}\}$ ($b$ is a Byzantine process). We can see that $S_B^*=\{v,v'\}$. Since $A_B\varsubsetneq S_B$, we have: $v\notin A_B$ or $v'\notin A_B$. Consider now the following configuration $\rho_0$: $prnt_r=prnt_b=\bot$, $level_r=level_b=0$, $prnt$ and $level$ variables of other processes are arbitrary (see Figure \ref{fig:impTAstrong}, other variables may have arbitrary values but other variables of $b$ are identical to those of $r$).

Assume now that $b$ takes exactly the same actions as $r$ (if any) immediately after $r$. Then, by symmetry of the execution and by convergence of $\mathcal{P}$ to $spec$, we can deduce that the system reaches in a finite time a configuration $\rho_1$ (see Figure \ref{fig:impTAstrong}) in which: $prnt_r=prnt_b=\bot$, $prnt_u=prnt_{u'}=r$, $prnt_v=prnt_{v'}=b$, $level_r=level_b=0$ and $level_u=level_{u'}=level_v=level_{v'}=1$ (because this configuration is the only one in which all correct process $v$ satisfies $spec(v)$ when $prnt_r=prnt_b=\bot$ and $level_r=level_b=0$). Note that $\rho_1$ is $A_B$-legitimate for $spec$ and $A_B$-stable (whatever $A_B$ is).

Assume now that $b$ behaves as a correct process with respect to $\mathcal{P}$. Then, by convergence of $\mathcal{P}$ in a fault-free system starting from $\rho_1$ which is not legitimate (remember that a strongly-stabilizing protocol is a special case of self-stabilizing protocol), we can deduce that the system reaches in a finite time a configuration $\rho_2$ (see Figure \ref{fig:impTAstrong}) in which: $prnt_r=\bot$, $prnt_u=prnt_{u'}=r$, $prnt_v=u$, $prnt_{v'}=u'$, $prnt_b=v$ (or $prnt_b=v'$), $level_r=0$, $level_u=level_{u'}=1$ $level_v=level_{v'}=2$ and $level_b=3$. Note that processes $v$ and $v'$ modify their O-variables in the portion of execution between $\rho_1$ and $\rho_2$ and that $\rho_2$ is $A_B$-legitimate for $spec$ and $A_B$-stable (whatever $A_B$ is). Consequently, this portion of execution contains at least one $A_B$-TA-disruption (whatever $A_B$ is).

Assume now that the Byzantine process $b$ takes the following state: $prnt_b=\bot$ and $level_b=0$. This step brings the system into configuration $\rho_3$ (see Figure \ref{fig:impTAstrong}). From this configuration, we can repeat the execution we constructed from $\rho_0$. By the same token, we obtain an execution of $\mathcal{P}$ which contains $c$-legitimate and $c$-stable configurations (see $\rho_1$) and an infinite number of $A_B$-TA-disruption (whatever $A_B$ is) which contradicts the $(t,A_B,1)$-TA strong stabilization of $\mathcal{P}$.
\end{proof}

\begin{figure}[t]
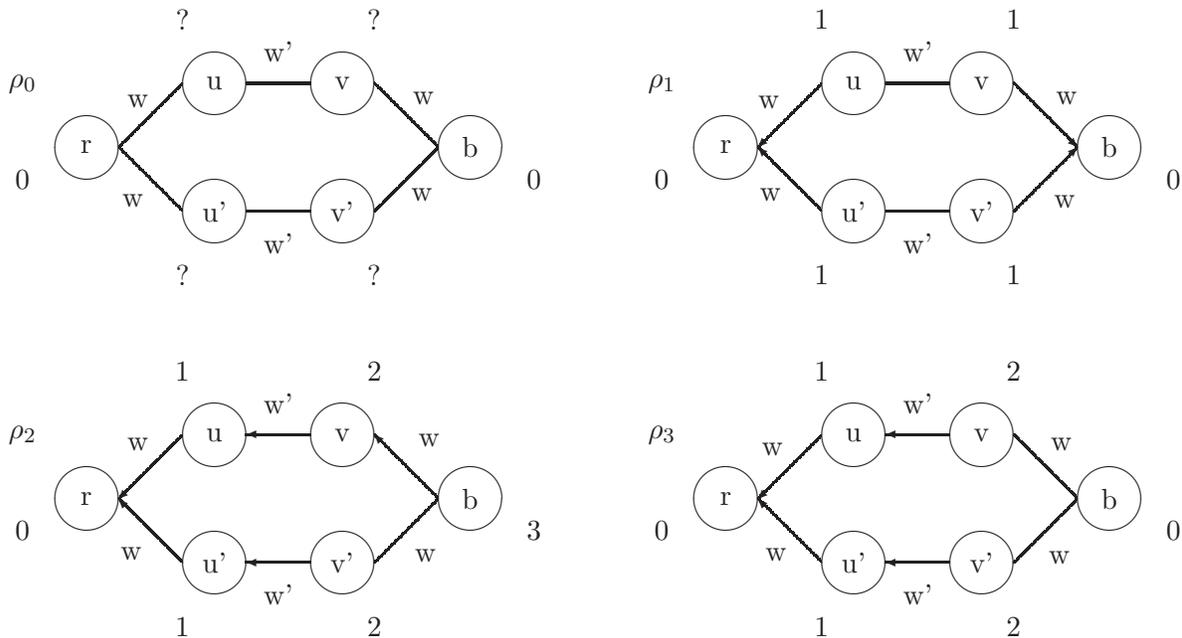

\noindent \begin{centering} \include{impTAstrong}
  \par\end{centering}
 \caption{Configurations used in proof of Theorem \ref{th:impTAStrong}.}
\label{fig:impTAstrong}
\end{figure}

\section{Conclusion}

In this article, we are interested in the BFS spanning tree construction in presence of both systemic transient faults and permanent Byzantine failures. As this task is global, it is impossible to solve it in a strictly stabilizing way. We proved then that there exists no solution to this problem even if we consider the weaker notion of strong stabilization.

Then, we provide a study of Byzantine containment properties of the well-known $min+1$ protocol. This protocol is one of the simplest self-stabilizing protocols which solve this problem. However, we prove that it achieves optimal area containment with respect to the notion of topology-aware strict and strong stabilization. All our results are summarized in the above table.

\footnotesize
\begin{center}
\begin{tabular}{|c||c|}
\cline{2-2}
\multicolumn{1}{c||}{}& BFS spanning tree construction \tabularnewline
\hline
\hline
$(c,f)$-strict stabilization & Impossible\tabularnewline
(for any $c$ and $f$)& (Theorem \ref{th:impStricte})\tabularnewline
\hline
$(t,c,f)$-strong stabilization  & Impossible\tabularnewline
(for any $t$, $c$, $f$)&(Theorem \ref{th:impStrong})\tabularnewline
\hline
$(A_B,f)$-TA strict stabilization  & Impossible\tabularnewline 
(for any $f$ and $A_{B}\varsubsetneq S_B$)& (Theorem \ref{th:impTAStricte})\tabularnewline
\hline
$(S_B,f)$-TA strict stabilization & Possible\tabularnewline 
(for $0\leq f\leq n-1$)& (Theorem \ref{th:possTAStricte})\tabularnewline
\hline
$(t,A_B,f)$-TA strong stabilization & Impossible\tabularnewline 
(for any $f$, $t$ and $A_{B}\varsubsetneq S_B^*$)& (Theorem \ref{th:impTAStrong})\tabularnewline
\hline
$(t,S_B^*,f)$-TA strong stabilization & Possible\tabularnewline 
(for $0\leq f\leq n-1$) & (Theorem \ref{th:possTAStrong}, $t=2m$)\tabularnewline
\hline
\end{tabular}
\end{center}

\normalsize

Using the result of \cite{DT01jb} about $r$-operators, we can easily extend results of this paper to some others problems as depth-search or reliability spanning trees. This work raises the following open questions. Has any other global static task as leader election or maximal matching a topology-aware strictly or/and strongly stabilizing solution ? We can also wonder about non static tasks as mutual exclusion (recall that local mutual exclusion has a strictly stabilizing solution provided by \cite{NA02c}).

\bibliographystyle{plain}
\bibliography{biblio}

\end{document}